\documentclass{article}

\usepackage[authoryear]{natbib} 
\usepackage[USenglish]{babel}
\usepackage[latin1]{inputenc}
\usepackage{indentfirst}
\usepackage{amsmath,amsfonts,amssymb,latexsym,amscd,exscale}
\usepackage{amsthm}
\usepackage{enumerate}
\usepackage{bm}
\usepackage{graphics,color}
\usepackage{fancyhdr}
\usepackage{authblk}
\usepackage{subcaption}
\usepackage{mathtools}
\usepackage{threeparttable}  
\usepackage{tabu,multirow}   

\newtheorem{prop}{Proposition}
\newtheorem{defn}{Definition}

\newcommand{\bt}{{\scriptsize \bm{\theta}}}

\title{Bayesian influence diagnostics using normalizing functional Bregman divergence}
\author[1]{Ian M Danilevicz}
\author[2]{Ricardo S Ehlers}

\affil[1]{Department of Statistics, Federal University of Minas Gerais, Brazil}
\affil[2]{Department of Applied Mathematics and Statistics, University of S\~ao Paulo, Brazil}

\date{}

\begin{document}

\thispagestyle{empty}
\maketitle
    
\begin{abstract}

Ideally, any statistical inference should be robust to local
influences. Although there are simple ways to check about leverage
points in independent and linear problems, more complex models require
more sophisticated methods. Kullback-Leiber and Bregman divergences
were already applied in Bayesian inference to measure the isolated
impact of each observation in a model. 
We extend these ideas to models for dependent data and with non-normal
probability distributions such as time series, spatial models and
generalized linear models. We also propose a strategy to rescale the
functional Bregman divergence to lie in the (0,1) interval thus
facilitating interpretation and comparison. This is accomplished with a 
minimal computational effort and maintaining all theoretical properties.
For computational efficiency, we take advantage of Hamiltonian Monte
Carlo methods to draw samples from the posterior distribution of model parameters.
The resulting Markov chains are then directly connected with Bregman
calculus, which results in fast computation. We check the propositions
in both simulated and empirical studies.      
  
\vspace{0.5 cm}
\textbf{Key-words}: Bayesian inference, functional Bregman divergence,
influential observations, Hamiltonian Monte Carlo. 
\end{abstract}

\section{Introduction}

After fitting a statistical model we need to investigate whether the
model assumptions are supported. In particular, inference about
parameters would be weak if it is influenced by a few individual results.
In this paper, we make use of a new diagnostic analysis tool for
detecting influential points. The idea is to
adapt the functional Bregman divergence to compare two or more
likelihoods (\cite{goh}) to the context of measuring how influent is each
observation in a given model. 
An influential point consists of an observation which strongly changes
the estimation of parameters. The classical example is a point which
drastically alters the slope parameter in a linear regression.  In
Bayesian inference, our focus lies into the whole posterior
distribution instead of a single parameter. 
Indeed, seeking for leverage effect in many parameters from a complex
model seems unfeasible.

Bayesian inference should produce a posterior distribution based on
Bayes theorem. So, if
there is a function which measures the distance between two
probability densities we can measure distance between two posterior
distributions or between a Bayesian model and its perturbed
version. The perturbed case could consist of the same sample without
an element if we have identical and independent observations
(\cite{goh}) or it may be a sample with an imputed element if we work
with dependent models (\cite{hao}).  
We can use a well know function such as Kullback-Leiber to measure the
divergence between two posterior distributions as well as
the functional Bregman divergence, which is a generalization of the previous
one. 

In the applications we have in mind, the posterior distributions are not
available in closed form and we resort to Markov chain Monte Carlo (MCMC)
methods to obtain approximations for parameter estimates and detection of
influential observations.
All the necessary computations in this paper were implemented using
the open-source statistical software {\tt R} (\cite{rproject}). In
particular, the {\tt rstan} package which is
an interface to the open-source Bayesian software Stan
(\cite{stanguide}) was used to draw samples from the joint
posterior distributions. Stan is a computing environment for
implementing Hamiltonian Monte Carlo methods (HMC, \cite{neal11})
coupled with the no-U-turn sampler (NUTS) 
which are designed to improve speed, stability and scalability
compared to standard MCMC schemes. Typically, HMC methods result in
high acceptance rates and low serial correlations thus leading to
efficient posterior sampling. 

The remainder of this paper is structured as follows. In Section
\ref{sec:models}, the models used to illustrate the application of our
propositions are briefly reviewed and the associated prior
distributions are described. The Hamiltonian Monte Carlo
sampling scheme is also described here. 
In Section \ref{sec:bregman} we introduce the functional Bregman divergence and
describe its use to detect influential observations in models for both
independent and dependent data.
Section \ref{sec:simulation} consists of simulation studies where we perform
sensitivity analysis and investigate how accurately we can detect influential observations.
Section \ref{sec:empirical} summarizes empirical studies in which we
illustrate our proposed methodology applied in real data. A discussion
in Section \ref{sec:discussion} concludes the paper. 

\section{Models}\label{sec:models}

This section is dedicated to describe the models which we used around
the paper and the HMC sampling scheme adopted.

\subsection{Generalized Linear Models}

Generalized linear models (GLM, \cite{nelder72}) are used here to
illustrate applications of our methods in models for dependent data
with non-normal distributions.
Let $y_1,\dots,y_n$ conditionally independent, where the
distribution of each $y_i$ belongs to the exponential family of distributions, i.e
\begin{equation}
f(y_i|\eta_i) = \exp \left\{ (\eta_i y_i -\psi(\eta_i)) + c(y_i) \right\}, ~i=1,\dots,n.
\label{glm}
\end{equation}  

\noindent
The density in equation (\ref{glm}) is parameterized by the canonical
parameter $\eta_i$ and $\psi(\cdot)$ and $c(\cdot)$ are known
functions. Also, $\bm{\eta}=(\eta_1, \dots,\eta_n)$ is related to 
regression coefficients by a monotone differentiable link function such that $g(\mu_i)=\eta_i$. 
The linear predictor is $\bm{\eta} = X \bm{\beta}$, where $X$ is the
design matrix and $\bm{\beta}=(\beta_1, \dots, \beta_k)$ is a $k$
vector of regression coefficients.  

The likelihood function based on model (\ref{glm}) is given by,
\begin{equation*}
L(\bm{y},\bm{\eta}) = \prod_{i=1}^{n} \exp \left\{ (\eta_i y_i -\psi(\eta_i)) + c(y_i) \right\}.
\end{equation*}

\noindent
This class of models includes several well kown distributions such as
Poisson, Binomial, Gamma, Normal and inverse Normal. 

\subsection{Spatial Regression Models}

We chose to illustrated our methods using
spatial regression models (SRM) as a kind of geostatistical data model (\cite{Gaetan}). 
The model can be represented as,
\begin{equation}
z_i = \beta_0 + \beta_1 x_i + \beta_2 y_i + \beta_3 x_i y_i + \beta_4 x_i^2 + \beta_5 y_i^2 + \varepsilon_i,
\label{srm_model}
\end{equation}

\noindent
where $z_i$ is the response of the observation $i$, $x_i$ is the value
of $z_i$ at the coordinate x, $y_i$ is the value of $z_i$ at the
coordinate y and $\varepsilon_i$ is an error, usually assumed
$N(0,1)$. Most commonly, x and y are latitude and longitude however
they could express as angles.
If we assume normality of the errors the likelihood function can be expressed as follows,
\begin{equation*}
L(\bm{\theta}, \bm{z}) = 
(2\pi)^{-n/2}
|\Sigma|^{1/2}\exp\left\{-\dfrac{1}{2}(\bm{z}-X\bm{\beta})'\Sigma(\bm{y}-X
\bm{\beta})\right\}, 
\label{like_srm}
\end{equation*}

\noindent
where $X$ is a design matrix with the following  columns: ones,
coordinate x, coordinate y, interaction of x and y, squared x and
squared y. The matrix $\Sigma$ describes  the covariance structure between the
observations and $\bm{y}$ in the response vector. 
A suitable set of priors consists of assuming that $\bm{\beta}\sim N(0,\eta I_6)$ and
the variance-covariance matrix follows a inverse Wishart
$\Sigma \sim IW(V,k)$.

A particular case of SRM consists of assuming an independence variance
structure, i.e. $\Sigma = (\sigma_1^2, \dots, \sigma_n^2)'I_n$, where
$I_n$ is an identity matrix. In this case default prior distributions for the $\sigma_i^2$
could be Inverse Gamma, or Gamma distributions if we are not restricted to conjugate priors.

\subsection{GARCH Model}

The generalized autoregressive conditional heteroscedasticity (GARCH)
model (\cite{boller86}) is the most used class of models to study the
volatility in financial markets.
The GARCH($p,q$) model is typically presented as the following sequence of equations,
\begin{eqnarray*}
y_t &=&  \sigma_t \epsilon_t , \; \epsilon_t \sim N(0, 1),\\
\sigma^2_t &=& \alpha_0 + \sum_{i=1}^{p} \alpha_i y_{t-i}^2 + \sum_{j=1}^{q} \beta_j \sigma^2_{t-j},
\end{eqnarray*}

\noindent
where $y_t$ is the observed return at time $t$ and
$\alpha_i$ and $\beta_j$ are unknown parameters. The $\epsilon_t$ are independent and identically
distributed error terms with mean zero and variance one.
Also, $\alpha_0>0$, $\alpha_i\ge 0, i=1,\dots,p$
and $\beta_j\ge 0, j=1,\dots,q$ define the positivity constraints and 
$\sum_{i=1}^p\alpha_i+\sum_{j=1}^q\beta_j < 1$, 
ensures covariance stationarity of $\sigma_t^2$.

Given an observed time series of returns $\bm{y}=\{y_1,\dots,y_n\}$
the conditional likelihood function is given by, 
\begin{equation}
L(\bm{\theta}) =
\prod_{t=s+1}^{n} \dfrac{1}{\sqrt{2\pi}\sigma_t}\exp\left\{ -\dfrac{y^2_t}{2 \sigma^2_t} \right\},  
\end{equation}

\noindent
where $s=\max(p,q)$ and $\bm{\theta}$ represents the set of all model
parameters. In practice, to get this recursive definition of the
volatility off the round in {\tt Stan} we need to impute non-negative
initial values for $\sigma$.

Prior distributions for the GARCH parameters were proposed by
\cite{prior_deschamps} and also used in \cite{prior_ardia}, who 
suggest a multivariate Normal distribution for $\bm{\alpha}$ and
$\bm{\beta}$ truncated to satisfy the associated constraints. However,
to avoid truncation we propose a simpler approach and specify the
following priors,
$\alpha_0\sim\mbox{Gamma}(a_0,b_0)$, $\alpha_i\sim\mbox{Beta}(c_i,
d_i)$ and $\beta_j\sim\mbox{Beta}(e_j,f_j)$ for $i\in \{1, \dots, p\}$
and $j \in \{1, \dots, q\}$ respectively.

\subsection{Hamiltonian Monte Carlo}

Our approach to detect influential observations relies on MCMC methods
that should produce Markov chains which efficiently explore the parameter
space. This motivates seeking for sampling strategies that aim at
reducing correlation within the chains thus improving convergence to
the posterior distribution.
Hamiltonian Monte Carlo (HMC) comes as a recent and powerful
simulation technique when all the parameters of interest are
continuous. 
HMC uses the gradient of the log posterior density to guide the
proposed jumps in the parameter space and reduces the random walk
effect in the traditional Metropolis-Hastings algorithm (\cite{Duane}
and \cite{neal11}). 

For $\bm{\theta}\in\mathbb{R}^d$ a $d$-dimensional vector of parameters
and $\pi(\bm{\theta})$ denoting the posterior density of $\bm{\theta}$,
the idea is to augment the parameter space whereas the invariant
distribution is now a Hamiltonian density given by,
\begin{equation*}
p(H(\bm{\theta},\bm{\varphi})) = \dfrac{1}{c} \exp(-H(\bm{\theta},\bm{\varphi})),
\end{equation*}
for a normalizing constant $c$. The Hamiltonian function is decomposed as,
$H(\bm{\theta},\bm{\varphi}) = U(\bm{\theta}) + K(\bm{\varphi})$,
where $U(\bm{\theta})$ is the potential energy,
$\bm{\theta}\in\mathbb{R}^d$ is the position vector, $K(\bm{\varphi})=\bm{\varphi}'V^{-1}\bm{\varphi}$
is the kinetic energy and
$\bm{\varphi}\in\mathbb{R}^d$ is the momentum vector in the physics literature. 
In a Bayesian setup we set $U(\bm{\theta})=-\log\pi(\bm{\theta})$.

Trajectories between points $(\bm{\theta},\bm{\varphi})$ are defined
theoretically by some differential equations which in practice cannot
be solved analytically. So, in terms of simulation a method is
required to approximately integrate the Hamiltonian dynamics. The
leapfrog operator (\cite{leimr04}) is typically used to discretize the
Hamiltonian dynamics and it updates $(\bm{\theta},\bm{\varphi})$ at time
$t+\epsilon$ as the following steps,
\begin{subequations}
\begin{align*}
  \bm{\varphi}^{(t)} &\sim N_d(0,V) \\
  \bm{\varphi}^{(t+\epsilon/2)}&=\bm{\varphi}^{(t)}-\dfrac{\epsilon}{2}\nabla_{\bt} U(\bm{\theta}^{(t)})\\
  \bm{\theta}^{(t+\epsilon)}  &=\bm{\theta}^{(t)}+\epsilon V^{-1}\bm{\varphi}^{(t+\epsilon/2)}\\
  \bm{\varphi}^{(t+\epsilon)} &=\bm{\varphi}^{(t+\epsilon/2)}-\dfrac{\epsilon}{2}\nabla_{\bt} U(\bm{\theta}^{(t+\epsilon)}),
  \end{align*}
\end{subequations}

\noindent
where $\epsilon>0$ is a  user specified small step-size and 
$\nabla_{\bt}U(\bm{\theta})$ is the gradient of
$U(\bm{\theta})$ with respect to $\bm{\theta}$. Then, after a given
number $L$ of time steps this results in a proposal
$(\bm{\theta}^*,\bm{\varphi}^*)$ and a Metropolis
acceptance probability is employed to correct the bias introduced by
the discretization and ensure convergence to the
invariant posterior distribution.

So, using Hamiltonian Monte Carlo involves specifying the number of
leapfrogs $L$ by iteration, the step-size length $\epsilon$ and the
initial distribution of the auxiliary variable $\bm{\varphi}$.  
The choice of an appropriated $L$ which associated with $\epsilon$
will not produce a constant periodicity may be done using the
No-U-Turn sampler (NUTS, \cite{hoffman14}), which aims at
avoiding the need to hand-tune $L$ and
$\epsilon$ in practice. During the warmup the algorithm will test different values
of leapfrogs and step-size and automatically judges the best range to sample. 
The basic strategy is to double $L$ until increasing the leapfrog will
no longer enlarge the distance between an initial value of $\bm{\theta}$ and a
proposed value $\bm{\theta}^*$. The criterion is the derivative with
respect to time of half the squared distance between the $\bm{\theta}$ and
$\bm{\theta}^*$. 
To define an efficient value of $\epsilon$, NUTS constantly checks if
the acceptance rate is sufficiently high during the warmup. If it is
not, the algorithm just shortens the step-size at next iteration (see
\cite{nesterov} and \cite{hoffman14}).

It is also worth noting that the Stan programming language provides a
numerical gradient using reverse-mode algorithmic differentiation so that
obtaining the gradient analytically is not necessary.
Finally, the distribution of $\bm{\varphi}$ is a multivariate
normal with either a diagonal or a full variance-covariance
matrix. The former is usually selected because the precision increase is
almost irrelevant compared to the computational memory costs (\cite{stanguide}). 
  
\section{Functional Bregman divergence}\label{sec:bregman}

The functional Bregman divergence aims at measuring dissimilarities
between functions, and in particular we are interested in comparing
posterior distributions. The method is briefly described here and
adapted to our models for detection of influential observations.
We define $(\Omega, X, \nu)$ as a finite measure space and $f_1(x)$
and $f_2(x)$ as two non-negative functions.  

\begin{defn}
Let us consider $\psi:(0,\infty) \rightarrow  \mathbb{R}$ being a
strictly convex and differentiable function on $\mathbb{R}$. Then 
the functional Bregman divergence $D_\psi$ is defined under the marginal density $\nu(x)$ as
\begin{equation}
D_\psi (f_1,f_2) = \int \psi(f_1(x)) - \psi(f_2(x)) - \psi'(f_2(x)) [f_1(x)-f_2(x)]   d\nu(x),
\label{bregman1}
\end{equation} 
where $\psi'$ represents the derivative of $\psi$.
\label{bregman2}
\end{defn}

\noindent This divergence has some well-known properties (see for
example \cite{goh}), the proofs of which
appear in \cite{frigyik,frigyik_B}.

Clearly, if $\psi(f(x))=f(x)$ $\forall f(x)$ then
$\psi'(f(x))= 1$ $\forall f(x)$ and the functional Bregman divergence is
zero for any $f_1(x)$ and $f_2(x)$. However, if we choose a strictly
convex $\psi$ then the Bregman divergence will be always
greater than zero, except for the trivial case
$f_1(x)=f_2(x)$. Furthermore, $\psi$ works as a tuning parameter and
increasing its distance from the identity we would have
$D_\psi(f_1,f_2)$ as large as desired no matter the functions $f_1(x)$
and $f_2(x)$. 
In this paper, we follow the suggestion in \cite{goh} and restrict to
the class of convex functions defined by \cite{eguchi}, 
\begin{align}
\psi_\alpha(x) = \left\{  \begin{array}{ll} 
 x\log x -x +1, & \alpha=1 \\
-\log x + x - 1, & \alpha = 0 \\
(x^{\alpha}-\alpha x + \alpha -1 ) / (\alpha^2 -\alpha ), & \textrm{otherwise.} 
\end{array} \right.
\label{xlogx}
\end{align}

\noindent Three popular choices of $\alpha$ are: $\alpha = 0$
(Itakura-Saito distance), $\alpha = 1$ (Kullback-Leibler divergence),
and $\alpha = 2$ (squared Euclidean distance or $L^2/2$). 

\subsection{Perturbation in dependent models}

Here we extend the ideas in \cite{goh} where perturbation was defined
in models for independent and identically distributed
observations to dependent models. A general perturbation is defined as
the ratio of unnormalized posterior densities, 
\begin{equation}
\delta(\bm{\theta}, \bm{y}, X) = 
\dfrac{f_\delta(\bm{y}|\bm{\theta},X)\pi_\delta(\bm{\theta})}{f(\bm{y}|\bm{\theta},X)\pi(\bm{\theta})},
\label{perturbation}
\end{equation}

\noindent
where $\delta$ indicates that likelihood and/or prior suffers some perturbation. 
Particularly, to assess potential influence of any observation the
perturbation is restricted to the likelihood function while keeping
the prior unaltered. The associated perturbation is then given by,
\begin{equation*}
\delta(\bm{\theta}, \bm{y}, X) =
\dfrac{f(\bm{y}_{(i)}|\bm{\theta},X)}{f(\bm{y}|\bm{\theta},X)},
\label{plike}
\end{equation*}

\noindent
where $\bm{y}_{(i)}$ denotes the vector $\bm{y}$ without the $i$th
case. In models for dependent data however we can not exclude an
observation without modifying the likelihood structure. 
In any case, the general rule to measure the local influence of the
$i$th point is to compute the divergence between
$f(\bm{y}|\bm{\theta},X)$ and $f(\bm{y}_{(i)}|\bm{\theta},X)$, 
\begin{equation*}
 d_{\psi,i}=D_\psi(f(\bm{y}_{(i)}|\bm{\theta},X),f(\bm{y}|\bm{\theta},X)).
 \label{di}  
\end{equation*}

The integral in Equation (\ref{bregman1}) however is analytically
intractable for most practical situations and an approximation is
needed. It is convenient to define the normalizing constant for $p(\bm{\theta}|\bm{y})$ as,
\begin{equation*}
m^{-1}(\bm{y}) = \int \dfrac{\omega(\bm{\theta})}{f(\bm{y}|\bm{\theta})p(\bm{\theta)}}p(\bm{\theta}|\bm{y})d\bm{\theta},
\label{constant}
\end{equation*}

\noindent
where $m(\bm{y})$ is the marginal density $\int f(\bm{\theta}|\bm{y})
p(\bm{\theta})$ and $\omega(.)$ is any probability density
function. So, given a sample $\{\bm{\theta}^s\}_{s=1}^S$ from
the posterior distribution (which could be generated by HMC) we can
estimate the normalizing constant as,
\begin{equation*}
\tilde{m}^{IW}(\bm{y}) = \left[ \dfrac{1}{S} \sum_{s=1}^{S}
  \dfrac{\omega(\bm{\theta}^s)}{f(\bm{y}|\bm{\theta}^s)p(\bm{\theta}^s)}
  \right]^{-1}. 
\label{mIW}
\end{equation*}

\noindent This is the so called Importance-Weighted Marginal Density
Estimate (IWMDE, \cite{chen94}). Denoting the resulting posterior
distribution as $\tilde{p}^{IW}(\bm{\theta}|\bm{y})$, the approximate
perturbed posterior is given by,
\begin{equation}
\tilde{p}_{\delta}^{IW}(\bm{\theta}|\bm{y}) = 
\dfrac{\tilde{p}^{IW}(\bm{\theta}|\bm{y}) \delta(\bm{\theta},\bm{y},X)}{\dfrac{1}{S}\sum_{s=1}^{S} \delta(\bm{\theta}^s,\bm{y},X)},
\end{equation}

\noindent
where $\tilde{p}^{IW}(\bm{\theta}|\bm{y}) =
f(\bm{y}|\bm{\theta})\pi(\bm{\theta})/\tilde{m}^{IW}(\bm{y})$. Consequently,
we can approximate the functional Bregman divergence between
$p(\bm{\theta}|\bm{y})$ and $p_{\delta}(\bm{\theta}|\bm{y})$ by, 
\begin{equation*}
\hat{D}_{\psi}^{IW} = \dfrac{1}{S} \sum_{s=1}^{S} \left\{ \dfrac{\psi(\tilde{p}^{IW}(\bm{\theta}^s|\bm{y}))-\psi(\tilde{p}_{\delta}^{IW}(\bm{\theta}^s|\bm{y}))
-(\tilde{p}^{IW}(\bm{\theta}^s|\bm{y}) -\tilde{p}_{\delta}^{IW}(\bm{\theta}^s|\bm{y})\psi'(\tilde{p}_{\delta}^{IW}(\bm{\theta}^s|\bm{y}))
}{\tilde{p}^{IW}(\bm{\theta}^s|\bm{y})} \right\},
\label{DIW}
\end{equation*}

\noindent
which for the convex functions in (\ref{xlogx}) is simplified as, 
\begin{equation*}
\hat{D}_{\psi_\alpha}^{IW} = \dfrac{1}{S} \sum_{s=1}^{S} \left\{
\dfrac{1-\alpha \{\delta(\bm{\theta}^s, \bm{y}, X)/\bar{\delta}\}^{\alpha-1} 
+(\alpha-1)\{\delta(\bm{\theta}^s, \bm{y}, X)/\bar{\delta}\}^{\alpha}} 
{\alpha (\alpha - 1) \{\tilde{p}^{IW}(\bm{\theta}^s|\bm{y})\}^{1-\alpha}} \right\},
\end{equation*}

\noindent
where $\bar{\delta}  = \dfrac{1}{S}\sum_{s=1}^{S}
\delta(\bm{\theta}^s,\bm{y},X)$. In particular, for $\alpha=1$ which
corresponds to the Kullback-Leibler divergence, we can simplify many
terms of the above expression and obtain,
\begin{equation}
\hat{D}_{\psi}^{IW} = \dfrac{1}{S} \sum_{s=1}^{S} \left\{
-\log \left( \dfrac{\delta(\bm{\theta}^s,\bm{y},X)}{\bar{\delta}} \right)
 \right\}.
\label{KL}
\end{equation}
 
\subsection{Normalizing Bregman divergence}

When using a functional Bregman divergence to evaluate influential
points each $d_{\psi,i}\in\mathbb{R}^+$ and in this scale we might have
doubts about one or more values being substantially higher than the others.
To facilitate comparison, \cite{mcculloch}
proposed a calibration which compresses the scale
between 0.5 and 1 by making an analogy with the comparison between two Bernoulli
distributions one of which with success probabilities equal to 1/2.
However, extending this idea to any functional Bregman divergence and
comparing any probability distribution with a Bernoulli sounds
difficult to justify theoretically.
Therefore, we propose a different route to compare the Bregman
divergence between two densities, which we call a normalizing Bregman divergence. 

\begin{prop}
  Given n+1 probability functions $f_0, \dots, f_n$ we have n
  divergences between $f_0$ and $f_1,\dots,f_n$, which we write as
  $D_\psi(f_0,f_1), \dots, D_\psi(f_0,f_n)$. Then, there is a function
  $||\psi||$ for which the sum of the $n$ divergences
  $D_{||\psi||}(f_0,f_1), \dots, D_{||\psi||}(f_0,f_n)$ returns one, 
  \begin{equation*}
    \sum_{i=1}^{n} D_{||\psi||}(f_0,f_i) = 1, \;
    \forall i \in \{1,\dots,n\},
  \end{equation*}
  and $D_{||\psi||}(\cdot)$ is called a normalizing Bregman divergence.
  \label{norm_berg}
\end{prop}

\begin{proof}
There is a sequence of functions $\psi_m: (f_1,f_2) \rightarrow
\mathbb{R}^+, \; m \in \mathbb{N}$, which tunes the divergence
intensity between any two density functions $f_1$ and $f_2$.   	
Suppose we have a full probability density $f_0$ and we wish to
compare it with each likelihood without $i$th element as $f_1, \dots,
f_n$ to check local influence. We already know that each divergence is
positive, so the sum of $n$ divergences belongs to positive real
domain, 
\begin{equation*}
\sum\limits_{i=1}^{n} D_{\psi}(f_0,f_i) = k_0, \; k_0 \in \mathbb{R}^+,
\label{berg_sum}
\end{equation*}
\noindent
where $k_0=0$ if and only if $\psi$ is the identity, but also could be
arbitrarily high as $\psi$ becomes more and more convex. In particular $k_0$ may be one. 
\end{proof}

\begin{prop}
Given n+1 probability functions $f_0, \dots, f_n$ we have n divergences between $f_0$ and $f_1,\dots,f_n$, which we write as $D_\psi(f_0,f_1), \dots, D_\psi(f_0,f_n)$. There is a mean operator $\mathcal{B}$ which transforms any Bregman divergence in a normalizing Bregman divergence.
\begin{equation}
D_{||\psi||}(f_0,f_q) = \mathcal{B}(D_{\psi}(f_0,f_q)) = \dfrac{D_{\psi}(f_0,f_q)}{\sum_{i=1}^{n} D_{\psi}(f_0,f_i)}, \;
\forall i,q \in \{1,\dots,n\},
\end{equation}
where $\mathcal{B}(\cdot)$ is called a normalizing Bregman operator.
\label{norm_berg_operator}
\end{prop}

\begin{proof}
By the generalized Pythagorean inequality (\cite{frigyik_B}), it is natural to supposed order maintenance as
$D_{\psi^*}(f_1,f_2) > D_{\psi^*}(f_1,f_3) \; \implies
D_{\psi^{**}}(f_1,f_2) > D_{\psi^{**}}(f_1,f_3)$ for
any $\psi^{*}$ and $\psi^{**}$ under the restriction of strictly convexity. 
If the above order relation is maintained then we can guarantee that all Bregman divergence with any $\psi$ consists of the same divergence just with a different location scale.  
\end{proof}

We gather these two arguments together as follows. A finite sum of divergences is
finite and $\psi$ just tunes the scale but not the order of
Bregman. So there is a special case of $\psi$, let us call it 
$||\psi||$, for which the sum with respect to a set of $n$ densities $f_i$ results in
one, and we call this divergence a normalizing Bregman
divergence. This is so because  all
Bregman divergence preserves the same order.

So, the attractiveness of our proposal is that
$0 \leq D_{||\psi||}(f_0,f_q) \leq 1, \; \forall
q \in \{1, \dots, n\}$ and it is quite intuitive to work in this scale
to compare divergences in the context of identifying influential observations.
Also, one possible caveat is that a result being high or low would depend on
the sample size so that any cut-off point should take $n$ into account.
In this paper we argue that uUnder the null hypothesis that there is
no influential observation in the sample, a reasonable expected
normalizing Bregman would be $1/n$, i.e.
\begin{equation*}
E(D_{||\psi||}(f_0,f_i)) = \dfrac{1}{n}, ~i=1,\dots,n,
\end{equation*}

\noindent 
so that we expected each observation would present the same divergence.
This bound becomes our starting point to identify influential observations.
If any observation returns $D_{||\psi||}(f_0,f_i)>1/n$, then it is a
natural candidate to be an influential point, which we must investigate. 
This should be seen as a useful practical device to seek for influence
rather than a definitive theoretical constant which can separate
influential from not influential cases. Finding a better cut-off point
other than $1/n$ is still as a open problem for future research.
Finally, we note that using the Kullback-Leibler divergence which is approximated
using Equation (\ref{KL}) leads to faster computations.

\section{Simulation Study}\label{sec:simulation}

In this section, we assess the performance of the algorithms and
methods proposed by conducting a simulation study. In particular, we
verify whether reliable results are produced and which parameters are
the most difficult to estimate. We also check sensitivity to prior
specification and the performance for detecting influential
observations. We concentrate on the performance of posterior
expectations as parameter estimators using Hamiltonian Monte Carlo
methods and the Stan package. For all combinations of models and prior
distributions we generate $m=1000$ replications of data and the
performances were evaluated considering the bias and the square root
of the mean square error (SMSE), which are defined as,
\begin{align*}
\textrm{Bias} = \frac{1}{m}\sum_{i=1}^{m}\hat{\theta}^{(i)} - \theta, \qquad
\textrm{SMSE} = \sqrt{\frac{1}{m}\sum_{i=1}^{m}(\hat{\theta}^{(i)} - \theta)^2}
\end{align*}
where $\hat{\theta}^{(i)}$ denotes the point estimate of a parameter $\theta$
in the $i$th replication, $i=1,\ldots,m$. 
Finally, for each data set we generated two chains of 4000 iterations
using Stan and discarded the first 2000 iterations as burn-in.

\subsection{Performance for Estimation and Sensitivity Analysis}

We begin with a logistic regression with an intercept and two
covariates and simulate data using two parametric sets. The values of
the two covariates  $x_1$ and $x_2$ were generated independently
from a standard normal distribution. The first model (Model 1) has true
parameters given by $\beta_{0}=1.3$, $\beta_{1}=-0.7$, $\beta_{2}=0.3$ while for the
second one (Model 2) the true parameters were set to $\beta_{0}=-1.6$,
$\beta_{1}=1.1$, $\beta_{2}=-0.4$ and each model was tested for two
different sample sizes, $n=100$ and $n=300$. Finally, inspired by
\cite{prior_gelman}, we adopted three different prior distributions
for the coefficients $\beta_j$, $j=0,1,2$ as follows. 
Prior 1: $\beta_{j}\sim N(0,10^2)$, Prior 2: $\beta_{j}\sim\text{Cauchy}(0,10)$, 
Prior 3: $\beta_{j}\sim\text{Cauchy}(0,2.5)$.
The main results of this exercise
are summarized in Table \ref{sa_lr}. Model 1 with $n=100$ and Cauchy prior
presents less bias for most estimations, except for $\beta_2$. This
prior also leads to lowest SMSE for all parameters. When
we  observed the same set but with $n=300$, all estimations get better
and again the Cauchy prior provides the least biased estimation. The
mean SMSE falls from
around 0.100 to approximately 0.025 when the sample size incrases. For 
Model 2, the results are quite similar.

\begin{center} [ Table \ref{sa_lr} around here ] \end{center}

We now turn to the analysis of the spatial regression model given in
(\ref{srm_model}). Data from two models with an intercept and four
covariates were generated where the true coefficients are given by
$\bm{\beta}$=$(3,0.25,0.65,0.2,-0.3,-0.2)$ (Model 1) and 
$\bm{\beta}$=$(3,-0.1,-0.4,0.8,-0.3,0.35)$ (Model 2), both with
the same variance $\sigma^2=1$. Each one was tested for two different
sample sizes, $n=50$ and $n=200$. 
Both latitude and longitude were generated by independent standard normal distributions without truncation, as an hypothetical surface without borders.

Each model was estimated under three different prior specifications
for the coefficients $\beta_{j}$, $j=0,\dots,5$ and $\sigma^2$ as follows. 
Prior 1: $\beta_{j}\sim N(0,100^2)$, $\sigma^2\sim\mbox{Gamma}(2,0.1)$,
Prior 2: $\beta_{j}\sim N(0,10^2)$, $\sigma^2\sim\mbox{Gamma}(2,0.1)$,
Prior 3: $\beta_{j}\sim N(0,10^2)$, $\sigma^2\sim\mbox{Gamma}(0.1,0.1)$,
where Prior 1 is more flat than the others and follows the
suggestion in \cite{Chung2013}, Prior 3 is more informative following
\cite{stanguide} warnings to avoid eventual computational errors and
Prior 2 is an intermediate one. The main results of this exercise are
summarized in Table \ref{sa_srm}. In this table, 0.0000 means smaller than 0.0001.

In Model 1 with $n=50$ the priors present bias of the same order for
most parameters except for $\sigma$, where Prior 3 is the best, and
$\beta_3$, where Prior 3 is the worst. The SMSEs are quite similar
across all prior specifications. When $n$ increases there is some
changes in bias order: $\beta_0$ has the best performance with Prior
3, however $\beta_2$, $\beta_4$ and $\beta_5$ show a one order decrease with Prior 2, then this is the best prior.
For Model 2 and $n=50$, we see Prior 2 again with better bias results
for $\beta_0$ and $\beta_4$, but Prior 3 is better for estimating
$\sigma$. With $n=200$ the bias results show an advantage of Prior 1
to estimate $\beta_2$ and $\beta_4$, as well as Prior 2 is better to
estimate $\beta_3$ and $\beta_5$, and likewise Prior 3 for
$\beta_0$. However, the SMSEs were very similar, so that the differences between priors were not so relevant in this last set. Overall, Prior 2 presents the best results. 

\begin{center} [ Table \ref{sa_srm} around here ] \end{center}

Our last exercise concerns to GARCH(1,1) models where
we generate artificial time series with Normal errors and two
different sets of parameters: $\alpha_{0,a} = 0.5$,
$\alpha_{1,a}=0.11$, $\beta_{1,a}= 0.88$ and $\alpha_{0,b} = 1$,
$\alpha_{1,b}=0.77$, $\beta_{1,b}= 0.22$. However we propose to
estimate both cases with Normal and Student $t$ error terms, even
though all series were built using Normal errors. Replacing a Normal
by a Student $t$ is a commonly used strategy to control overdispersed data.
The prior distributions were assigned as follows.
Prior 1: $\alpha_{0}\sim\mbox{Gamma}(0.1,0.1)$, $\alpha_{1}\sim\mbox{Beta}(2,2)$, $\beta_{1}\sim\mbox{Beta}(2,2)$, 
Prior 2: $\alpha_{0}\sim\mbox{Gamma}(0.1,0.1)$, $\alpha_{1}\sim\mbox{Beta}(2,3)$, $\beta_{1}\sim\mbox{Beta}(3,2)$ and
Prior 3: $\alpha_{0}\sim\mbox{Gamma}(0.5,0.5)$, $\alpha_{1}\sim\mbox{Beta}(2,3)$, $\beta_{1}\sim\mbox{Beta}(3,2)$,


The results for the GARCH(1,1) with Normal errors are presented in
Table \ref{sa_garchn}. We notice that Prior 3 attained the best
results for the parameter set 1, but Prior 1 was better in set 2. This
outcome happens because Prior 1 is perhaps too informative about
$\alpha_0$ and even a value of $T$ as large as 900 was not enough for the
model to learn from data. However, the different priors assigned to
$\alpha_1$ and $\beta_1$ do not imply in any drastic output change. 
Table \ref{sa_garcht} summarizes the output from a GARCH(1,1) model estimated
with Student $t$ errors. 
From this table we notice that Prior 1 returned the best results for the
parameter set 2, but in set 1 the three priors share similar performances. 
We now look at both tables in tandem to compare Student $t$ and Normal errors.
The Normal GARCH presented better results than Student $t$ for the
parameter set 2, but they were similar in set 1, so that there is no
need of a robust model in this case (the series were generated with normal errors).

\begin{center} [ Table \ref{sa_garchn} around here ] \end{center}

\begin{center} [ Table \ref{sa_garcht} around here ] \end{center}

\subsection{Influence Identification}

To evaluate the normalizing Bregman divergence as a useful tool to
identify point influence we proceed with three simulation sets, each
refering to a different model. We use the same models presented in the
previous subsection. Within each model we created four scenarios: I
without any kind of perturbation, II where there is one perturbed
observation, III where there are two influential points and IV with
three influential points. The point contamination in time series and
spatial models follows the scheme proposed by \cite{hao} and
\cite{cho}, i.e., $y^*_t=y_t + 5\sigma_y$, where $\sigma_y$ is the
standard deviation of the observed sample $\bm{y}$. For the logistic
regression however we need a different approach to
contaminate data. In this case we simply exchange the output, i.e. if
$y_t$ is to be contaminated and $y_t=1$ then we set $y^*_t=0$,
otherwise if $y_t=0$ we set $y^*_t=1$.

The results for the case influence diagnostic using normalizing
Bregman divergence in logistic regression are shown in Table
\ref{ci_lr}. For this table, the true parameter values are
$\beta_0=-3$, $\beta_1=-0.7$, $\beta_2=0.3$ and the prior
distributions are $\beta_{j}\sim\text{Cauchy}(0,2.5)$,
~$j\in\{0,1,2\}$. Also, the perturbation schemes are: I no
perturbation, II observation 64 has an additional noise, III
observations 44 and 64 present perturbation and IV observations 19, 44
and 64 have an extra noise. The table then shows the estimated (mean
and standard deviation) divergences for the three observations, 19, 44
and 64.

We first notice that in the no perturbation scenario the estimated
divergences are mostly as expected on average, i.e. $1/100$ and
$1/300$ for $n=100$ and $n=300$ respectively. On the other hand, when
the output is perturbed the average divergence was between 0.028 and
0.031 for $n=100$ and between 0.008 and 0.009 for $n=300$.
Finally, there is a correlation between mean and standard deviation in
the sense that a small value of one corresponds to low estimation of the other.

\begin{center} [ Table \ref{ci_lr} around here ] \end{center}

The results are even more emphatic in spatial regression models as
shown in Table \ref{ci_srm}. In this table, the true parameter values
are $\phi=0.75$, $\sigma^2=1$, $\beta_0=1.3$, $\beta_1=-0.7$ and we
chose Prior 2. Also, the influence scenarios are: I without
perturbation, II observation 19 has an additional noise, III
observations 15 and 19 present perturbation and IV where 3, 15 and 19
have an extra noise. 

For $n=50$ and scenario I the normalizing Bregman divergence has mean
around 0.020 in the three observed points, which corresponds to the expected
$1/50$. In scenario II, the estimates for observations 3 and 15 fall to 0.011 and 0.013,
because observation 19 was perturbed and its divergence estimate rises
to 0.423. In scenario III the estimate for observation 3 falls even more, because both 15 and 19 were
perturbed and both have a 0.283 estimate for the divergence. Finally, when the three
observations were perturbed they share the impact between 0.210 and
0.214 estimates.
For $n=200$ and scenario I, we have again results precisely as
expected, i.e. 0.005 compared to $1/200$. Furthermore, scenarios II, III and IV are
quite similar, the mean values are slightly smaller, but this is
expected for a larger sample.  

\begin{center} [ Table \ref{ci_srm} around here ] \end{center}

The effect is still clear in time series with moderate sample sizes,
as shown in Table \ref{ci_garch}. This table shows results for a
GARCH(1,1) model with normal errors and true parameter values
$\alpha_0=2$, $\alpha_1=0.2$ and $\beta_1=0.6$. Influence scenarios
for $T=100$: I without perturbation, II observation 64 has an
additional noise, III observations 44 and 64 present perturbation and
IV where 19, 44 and 64 have an extra noise. Influence scenarios to
$T=500$: I without perturbation, II observation 464 has an additional
noise, III observations 344 and 464 present perturbation and IV where
119, 344 and 464 have an extra noise. 

For $T=100$ and scenario I the normalizing
Bregman divergence has mean equal to 0.009 in the three observed
points, which is slightly bellow the expected $1/100$. In scenario
II the estimated divergence for observations 19 and 44 fall to 0.007,
because observation 64 was perturbed and its estimated divergence
rises to 0.247, which might not seem a large
value but it is more than 20 times the expected value $1/100$. In III
the estimate for observation 19 falls even more, because both 44 and 64 were
perturbed and have estimated divergences 0.188 and 0.192. Finally, when all
three observations were perturbed they share the impact with similar
estimated divergences. For $T=500$ and scenario I, we have again
results around 0.002, i.e. $1/500$. Furthermore, scenarios II,
III and IV show quite similar results, the estimated values being
slightly smaller, but this is expected for a larger sample.  

\begin{center} [ Table \ref{ci_garch} around here ] \end{center}

\section{Empirical Analysis}\label{sec:empirical}

In this section, we investigate influential points in real data sets
using the normalizing Bregman divergence. In all examples, convergence
assessment of the Markov chains were based on visual inspection of
trace plots, autocorrelation plots and the $\hat{R}$ statistic since
we ran two chains for each case. All results indicated that the chains
reached stationarity relatively fast.

\subsection{Binary Regression for Alpine Birds}

A study about an endemic coastal alpine bird was conducted in Vancouver
Island (Southwest coast of British Columbia, Canada) for
more than a decade and the results were published in
\cite{Jackson15}. The presence or absence of birds in a grid of space
was registered over the years together with other environmental
characteristics as covariates.
The authors proposed an interpretation of data by a Random Forest
model. Here we extend their model to a Bayesian framework and consider
a binary logistic regression with other covariates.

For illustration, we selected the following covariates:
elevation (1000 meters) and average temperature in summer months (in
Celcius degrees) and the model also includes an intercept.
We then ran a HMC with two chains, each one with 4,000 iterations where the
first half was used as burn-in. This setup was used to fit models with
probit and logit link functions. The
normalizing Bregman divergences estimated for each observation are
displayed in Figure \ref{pombas_breg}. 
From this figure, it is hard to judge what is a high value of divergence,
because there are more than one thousand observations in the sample.  
However, we can easily conclude that the model with logit link
performs better because the highest values of logit are lower than the highest values
of probit. This is to say that the most influential points in the logit model are
not so influent as in the probit model.  
 
\begin{center} [ Figure \ref{pombas_breg} around here ] \end{center}

\subsection{Spatial models for rainfall in South of Brazil}

Here we illustrate a spatial regression approach to analyze the data
on precipitation levels in Paran\'a State, Brazil. This data is freely
available in the {\tt geoR} package and was previously analyzed by for
example \cite{Diggle2002} and \cite{Gaetan}. 
The data refers to average precipitation levels over 33 years of
observation during the period May-June (dry-season) in 143 recording
stations throughout the state.
The original variable was summarized in 100 millimeters of
precipitation per station. We changed it to 10,000 millimeters of rain
per station, which seems more intuitive once the average local
precipitation is around 1,000 millimeter per year and the observation
time was larger than 10 years. 
 
We the fitted three models for the average rainfall. These are the full SRM
presented in Equation (\ref{srm_model}), the same model but without the
squared components $x^2$ and $y^2$, and the smallest one without
squared components and neither the interaction term $x*y$, which we
refer to as full, middle and small models respectively.
All the models were fitted from two HMC chains, each one with 20,000
iterations where the first half was used as burn-in. 

We estimated the normalizing Bregman divergence for each recording
station and compared the results in the same way as in the previous example.
We conclude that the small model was the best one in the sense that
it shows the smallest peaks. For example, the maximum values for each
model were 0.072, 0.068 and 0.039 respectively, which are already
pretty high relative to the expected $1/143=0,0069$.

We chose to display only the results for this one
best model in Figure \ref{parana_map_breg} from which we can see that the
largest values of normalizing Bregman (largest circles) are scattered around the map,
notwithstanding the rainy region is concentrated in the southwest.
 
\begin{center} [ Figure \ref{parana_map_breg} around here ] \end{center}

\subsection{GARCH for Bitcoin exchange to US Dollar}

The cryptocurrencies were born in the new millennium dawn as an
alternative as governments and banks. As such, they changed the rules
of financial market and they appreciated very fast, although high
fluctuation and sharp falls are common. In particular, the Bitcoin is
likely the most famous cryptocurrency and shows the largest volume of crypt transactions.

We illustrate the statistical analysis with one year of daily data on
the log-returns of Bitcoin (BTC) exchange to U.S. Dollar (USD) from August 5, 2017 to August 5, 2018.
This data was produced from the CoinDesk price page (see
http://www.coindesk.com/price/).
We then fitted GARCH(1,1) models with normal and Student $t$ errors
for the log-return of BTC to USD exchange.
We ran the HMC with two chains, each one with 4,000 iterations where
the first half was used as burn-in.
The estimation of main parameter of Normal model are:  
$\alpha_0$ has zero mean and SD,
$\alpha_1$ is 0.15 (0.06) and
$\beta_1$ is 0.07 (0.08), on the other hand the Student $t$ is
$\alpha_0$ with zero mean and SD too,
$\alpha_1$ is 0.11 (0.05) and
$\beta_1$ is 0.06 (0.07).

We estimated the normalizing Bregman divergence for each day, the
result could be see in Figure \ref{btc}. Here is not so trivial to
choose between the Normal or the Student $t$ model. Because there is
no clear dominance of one or another, even though the Student $t$
presents the highest value of divergence, both form a mixed cloud of
values very closed. However the highest points are quite sure very
influent observations, because they represent more than 20 times the
expected mean of 1/364. 
Consequently, it is not a surprise that the observed high divergences
in January correspond to what the Consumer News and Business Channel
(CNBC) called a Bitcoin nightmare. A time of new regulations in South
Korea as well as a Facebook currency policy change, which implied a
devaluation. 
    
\begin{center} [ Figure \ref{btc} around here ] \end{center}

\section{Discussion}\label{sec:discussion}

In this article we explored the possibilities of using the functional
Bregman divergence as a useful generalization of the Kullback-Leiber
divergence to identify influential observations in Bayesian models,
for both dependent and independent data. 
Kullback-Leiber is easier to estimate, but overall it is
difficult to infer if a point represents an influential point or
not. So we propose to normalizing the Bregman divergence based on the
order maintenance of the functional. 
It has two intuitive
advantages: firstly, it belongs to range between zero and one which
is easier to interpret, secondly we can evaluate its intensity 
according to the sample size. 
In particular, the normalizing Bregman divergence for the  Kullback-Leiber case
avoids the need for heavy computations.

As we saw in the simulation study, the expected average of a normalizing Bregman
divergence to any observation without perturbation is approximately
$1/n$. Of course that number of influential points is a relevant issue
to evaluate the value of a normalizing Bregman divergence. 
The simulation study embraced three different fields of statistic:
GLM, spatial models and time series, with similar conclusions in all of them.
Besides, the empirical analysis explored three
scientific fields: Ecology, Climatology and Finance. 
Finally, in all cases the Hamiltonian Monte Carlo was an efficient 
and fast way to obtain samples from the posterior distribution of parameters

\section*{Acknowledgments}

Ricardo Ehlers received support from S\~ao Paulo Research Foundation
(FAPESP) - Brazil, under grant number 2016/21137-2.

\clearpage

\begin{table}[h]
	\centering
	\caption{Bias and square root of mean square error for
          parameter estimates in the logistic regression.}
\begin{tabular}{cc rcc rcc rc}
\hline \multicolumn{10}{c}{Model 1} \\	
	&	        & \multicolumn{2}{c}{Prior 1} & &
\multicolumn{2}{c}{Prior 2} & & \multicolumn{2}{c}{Prior 3} \\ \cline{3-4} \cline{6-7} \cline{9-10}
$n$ &     Parameter & Bias & SMSE && Bias & SMSE && Bias & SMSE \\
100 & $\beta_{0}$ &  0.125 & 0.105 & &  0.126 & 0.112  & &  0.096 & 0.092\\ 
    & $\beta_{1}$ & -0.084 & 0.113 & & -0.079 & 0.108  & & -0.060 & 0.095\\ 
    & $\beta_{2}$ &  0.040 & 0.088 & &  0.045 & 0.089  & &  0.095 & 0.083\\ 
300 & $\beta_{0}$ &  0.031 & 0.024 & &  0.031 & 0.026  & &  0.030 & 0.025 \\ 
    & $\beta_{1}$ & -0.021 & 0.027 & & -0.026 & 0.027  & & -0.015 & 0.024 \\ 
    & $\beta_{2}$ &  0.009 & 0.022 & &  0.010 & 0.024  & &  0.009 & 0.021 \\ 
\hline \multicolumn{10}{c}{Model 2} \\	
&	        & \multicolumn{2}{c}{Prior 1} & &
\multicolumn{2}{c}{Prior 2} & & \multicolumn{2}{c}{Prior 3} \\ \cline{3-4} \cline{6-7} \cline{9-10}    
$n$ &     Parameter & Bias & SMSE && Bias & SMSE && Bias & SMSE \\
100 & $\beta_{0}$ & -0.155 & 0.151 & & -0.171 & 0.172   & & -0.117 &  0.146  \\ 
    & $\beta_{1}$ &  0.145 & 0.164 & &  0.130 & 0.171   & &  0.094 &  0.131 \\ 
    & $\beta_{2}$ & -0.055 & 0.106 & & -0.061 & 0.109   & & -0.037 &  -0.037  \\ 
300 & $\beta_{0}$ & -0.056 & 0.038 & & -0.045 & 0.039 & & -0.026 & 0.036 \\ 
    & $\beta_{1}$ &  0.051 & 0.040 & &  0.033 & 0.041 & &  0.022 & 0.036 \\ 
    & $\beta_{2}$ & -0.011 & 0.027 & & -0.014 & 0.030 & & -0.008 & 0.028 \\ \hline
\end{tabular}
	\label{sa_lr}
\end{table}

\clearpage

\begin{table}
  \centering
  \caption{Bias and square root of mean square error for
          parameter estimates in the spatial regression model.}
  \begin{tabular}{cc rcc rcc rc}
\hline \multicolumn{10}{c}{Model 1} \\		
 & & \multicolumn{2}{c}{Prior 1} & & \multicolumn{2}{c}{Prior 2} & & \multicolumn{2}{c}{Prior 3} \\ \cline{3-4} \cline{6-7} \cline{9-10} 
$n$  & Parameter     & Bias   & SMSE  & & Bias    & SMSE  & & Bias   & SMSE \\
{50} & $\beta_0$ & -0.0026 & 0.0478 &  & -0.0126 & 0.0465 &  & -0.0024 & 0.0426 \\ 
& $\beta_1$ & 0.0024 & 0.0278 &  & 0.0013 & 0.0248 &  & -0.0031 & 0.0254 \\ 
& $\beta_2$ & 0.0016 & 0.0272 &  & 0.0007 & 0.0242 &  & -0.0063 & 0.0256 \\ 
& $\beta_3$ & 0.0035 & 0.0326 &  & 0.0014 & 0.0306 &  & -0.0133 & 0.0319 \\ 
& $\beta_4$ & 0.0018 & 0.0178 &  & 0.0055 & 0.0160 &  & 0.0018 & 0.0157 \\ 
& $\beta_5$ & -0.0038 & 0.0166 &  & 0.0020 & 0.0153 &  & -0.0023 & 0.0159 \\ 
& $\sigma$ & 0.0324 & 0.0137 &  & 0.0329 & 0.0135 &  & 0.0065 & 0.0117 \\ \hline
{200} & $\beta_0$ & -0.0022 & 0.0112 &  & -0.0050 & 0.0105 &  & -0.0006 & 0.0102 \\ 
& $\beta_1$ & 0.0021 & 0.0054 &  & -0.0029 & 0.0056 &  & 0.0033 & 0.0054 \\ 
& $\beta_2$ & -0.0056 & 0.0053 &  & -0.0003 & 0.0054 &  & 0.0029 & 0.0051 \\ 
& $\beta_3$ & 0.0067 & 0.0055 &  & -0.0034 & 0.0053 &  & -0.0017 & 0.0058 \\ 
& $\beta_4$ & -0.0008 & 0.0029 &  & 0.0000 & 0.0029 &  & 0.0001 & 0.0030 \\ 
& $\beta_5$ & -0.0026 & 0.0031 &  & 0.0001 & 0.0029 &  & 0.0017 & 0.0029 \\ 
& $\sigma$ & 0.0064 & 0.0026 &  & 0.0093 & 0.0027 &  & 0.0049 & 0.0025 \\ \hline
\multicolumn{10}{c}{Model 2} \\		
 & & \multicolumn{2}{c}{Prior 1} & & \multicolumn{2}{c}{Prior 2} & & \multicolumn{2}{c}{Prior 3} \\ \cline{3-4} \cline{6-7} \cline{9-10} 
{50} & $\beta_0$ & -0.0041 & 0.0474 &  & 0.0008 & 0.0459 &  & 0.0082 & 0.0488 \\ 
& $\beta_1$ & -0.0031 & 0.0258 &  & -0.0074 & 0.0265 &  & -0.0088 & 0.0271 \\ 
& $\beta_2$ & -0.0050 & 0.0283 &  & -0.0028 & 0.0266 &  & 0.0056 & 0.0269 \\ 
& $\beta_3$ & -0.0086 & 0.0313 &  & -0.0092 & 0.0335 &  & 0.0010 & 0.0307 \\ 
& $\beta_4$ & -0.0041 & 0.0169 &  & 0.0008 & 0.0167 &  & -0.0070 & 0.0167 \\ 
& $\beta_5$ & 0.0018 & 0.0161 &  & -0.0013 & 0.0173 &  & -0.0012 & 0.0169 \\ 
& $\sigma$ & 0.0344 & 0.0132 &  & 0.0266 & 0.0121 &  & 0.0099 & 0.0117 \\ \hline
{200} & $\beta_0$ & -0.0024 & 0.0101 &  & 0.0017 & 0.0108 &  & 0.0007 & 0.0105 \\ 
& $\beta_1$ & 0.0023 & 0.0053 &  & -0.0019 & 0.0052 &  & 0.0055 & 0.0048 \\ 
& $\beta_2$ & -0.0001 & 0.0054 &  & 0.0009 & 0.0054 &  & 0.0010 & 0.0053 \\ 
& $\beta_3$ & 0.0019 & 0.0058 &  & 0.0004 & 0.0057 &  & -0.0012 & 0.0060 \\ 
& $\beta_4$ & 0.0000 & 0.0027 &  & 0.0018 & 0.0028 &  & -0.0014 & 0.0030 \\ 
& $\beta_5$ & 0.0015 & 0.0027 &  & -0.0001 & 0.0028 &  & 0.0007 & 0.0029 \\ 
& $\sigma$ & 0.0053 & 0.0026 &  & 0.0084 & 0.0025 &  & 0.0037 & 0.0025 \\ \hline
\end{tabular}
\label{sa_srm}
\end{table}

\clearpage

\begin{table}
  \centering
  \caption{Bias and square root of mean square error for
    parameter estimates in the GARCH(1,1) model with normal errors.}
  \begin{tabular}{cc rcc rcc rc}
    \hline \multicolumn{10}{c}{Model 1} \\		
    & & \multicolumn{2}{c}{Prior 1} & & \multicolumn{2}{c}{Prior 2} & & \multicolumn{2}{c}{Prior 3} \\ \cline{3-4} \cline{6-7} \cline{9-10} 
    $T$ & Parameter     & Bias  & SMSE  & & Bias    & SMSE  & & Bias   & SMSE \\
500	& $\alpha_0$ & 	0.716	&	0.725	& &	0.339	&
0.317	& &	0.178	&	0.195	\\
    & $\alpha_1$ & 	-0.111	&	0.013	& &	-0.119	&	0.015	& &	-0.121	&	0.015	\\
    & $\beta_1$ & 	-0.174	&	0.037	& &	-0.095	&	0.015	& &	-0.065	&	0.009	\\
900	& $\alpha_0$ & 	0.773	&	0.837	& &	0.399	&	0.371	& &	0.215	&	0.235	\\
    & $\alpha_1$ & 	-0.131	&	0.018	& &	-0.136	&	0.019	& &	-0.136	&	0.019	\\
    & $\beta_1$ & 	-0.167	&	0.036	& &	-0.095	&	0.016	& &	-0.063	&	0.011	\\ 
	\hline \multicolumn{10}{c}{Model 2} \\		
	& & \multicolumn{2}{c}{Prior 1} & & \multicolumn{2}{c}{Prior 2} & & \multicolumn{2}{c}{Prior 3} \\ \cline{3-4} \cline{6-7} \cline{9-10} 
$T$   & Parameter     & Bias   & SMSE  & & Bias    & SMSE  & & Bias   & SMSE \\
500	& $\alpha_0$ & 	0.031	&	0.021	& &	-0.049	&	0.023	& &	-0.056	&	0.022	\\
    & $\alpha_1$ & 	-0.300	&	0.095	& &	-0.319	&	0.106	& &	-0.321	&	0.107	\\
    & $\beta_1$ & 	0.009	&	0.005	& &	0.067	&	0.010	& &	0.071	&	0.011	\\
900	& $\alpha_0$ & 	0.093	&	0.024	& &	0.033	&	0.015	& &	0.036	&	0.015	\\
    & $\alpha_1$ & 	-0.297	&	0.091	& &	-0.315	&	0.102	& &	-0.315	&	0.102	\\
    & $\beta_1$ & 	-0.029	&	0.005	& &	0.014	&	0.004	& &	0.012	&	0.004	\\ 
\hline
  \end{tabular}
  \label{sa_garchn}
\end{table}

\clearpage

\begin{table}
\centering
\caption{Bias and square root of mean square error for
  parameter estimates in the GARCH(1,1) model with Student $t$ errors.}
\begin{tabular}{cc rcc rcc rc}
\hline \multicolumn{10}{c}{Model 1} \\		
    & & \multicolumn{2}{c}{Prior 1} & & \multicolumn{2}{c}{Prior 2} & & \multicolumn{2}{c}{Prior 3} \\ \cline{3-4} \cline{6-7} \cline{9-10} 
$T$ & Parameter     & Bias   & SMSE  & & Bias   & SMSE  & & Bias   & SMSE \\		
500	& $\alpha_0$& 0.133  & 0.136 & & -0.191	&	0.132	& &	-0.262	&	0.147	\\
	& $\alpha_1$&-0.121  & 0.015 & & -0.128	&	0.017	& &	-0.129	&	0.017	\\
	& $\beta_1$ &-0.174  & 0.036 & & -0.093	&	0.013	& &	-0.078	&	0.010	\\
900	& $\alpha_0$& 0.192  & 0.152 & & -0.147	&	0.116	& &	-0.224	&	0.126	\\
	& $\alpha_1$&-0.139  & 0.020 & & -0.143	&	0.021	& &	-0.144	&	0.021	\\
	& $\beta_1$ &-0.171  & 0.035 & & -0.090	&	0.013	& &	-0.072	&	0.009	\\
\hline \multicolumn{10}{c}{Model 2} \\		
    & & \multicolumn{2}{c}{Prior 1} & & \multicolumn{2}{c}{Prior 2} & & \multicolumn{2}{c}{Prior 3} \\ \cline{3-4} \cline{6-7} \cline{9-10}	
$T$ & Parameter     & Bias   & SMSE  & & Bias    & SMSE  & & Bias   & SMSE \\
500	& $\alpha_0$& -0.218	&	0.060	& &	-0.291	&	0.096	& &	-0.276	&	0.087	\\
	& $\alpha_1$& -0.361	&	0.133	& &	-0.383	&	0.149	& &	-0.385	&	0.151	\\
	& $\beta_1$ &  0.033	&	0.006	& &	0.099	&	0.015	& &	0.093	&	0.013	\\
900	& $\alpha_0$& -0.154	&	0.032	& &	-0.219	&	0.057	& &	-0.214	&	0.054	\\
	& $\alpha_1$& -0.368	&	0.137	& &	-0.378	&	0.145	& &	-0.376	&	0.143	\\
	& $\beta_1$ & -0.009	&	0.003	& &	0.040	&	0.006	& &	0.036	&	0.005	\\ 
\hline
  \end{tabular}
  \label{sa_garcht}
\end{table}

\clearpage
 
\begin{table}
  \centering
  \caption{Case influence diagnostic in logistic regression by normalizing Bregman divergence.}
  \begin{tabular}{cccc cccc}
\hline  \multicolumn{8}{c}{$n=100$}\\ 
Perturbation & Obs & Mean & SD  & Perturbation & Obs  & Mean & SD \\ \hline
I	& 19    & 0.010    & 0.016   & II & 19    & 0.009     & 0.012   \\ 
I	& 44    & 0.009    & 0.016   & II & 44    & 0.009     & 0.009   \\ 
I	& 64    & 0.010    & 0.017   & II & 64    & 0.031     & 0.035   \\ 
III	& 19    & 0.009    & 0.012   & IV & 19    & 0.028     & 0.025   \\ 
III	& 44    & 0.028    & 0.027   & IV & 44    & 0.030     & 0.029   \\ 
III	& 64    & 0.030    & 0.034   & IV & 64    & 0.031     & 0.030   \\ \hline 
\multicolumn{8}{c}{$n=300$}\\ 
Perturbation & Obs & Mean  & SD & Perturbation & Obs & Mean  & SD \\ \hline
I	& 19    & 0.003    & 0.005   & II & 19    & 0.003     & 0.005   \\ 
I	& 44    & 0.003    & 0.004   & II & 44    & 0.003     & 0.003   \\ 
I	& 64    & 0.003    & 0.004   & II & 64    & 0.008     & 0.010   \\ 
III	& 19    & 0.003    & 0.004   & IV & 19    & 0.009     & 0.010   \\ 
III	& 44    & 0.008    & 0.010   & IV & 44    & 0.008     & 0.009   \\ 
III	& 64    & 0.008    & 0.008   & IV & 64    & 0.008     & 0.010   \\ \hline 
  \end{tabular}
  \label{ci_lr}
\end{table}

\clearpage

\begin{table}
  \centering
  \caption{Case influence diagnostic for spatial regression model
    using normalizing Bregman divergence.}
  \begin{tabular}{cccc cccc}
\hline
\multicolumn{8}{c}{$n=50$}\\ 
Perturbation & Obs & Mean & SD & Perturbation & Obs & Mean & SD \\ \hline
I & 3 & 0.020 & 0.027 & II & 3 & 0.011 & 0.016 \\ 
I & 15 & 0.018 & 0.027 & II & 15 & 0.013 & 0.019 \\ 
I & 19 & 0.018 & 0.026 & II & 19 & 0.423 & 0.111 \\ 
III & 3 & 0.009 & 0.013 & IV & 3 & 0.211 & 0.074 \\ 
III & 15 & 0.283 & 0.089 & IV & 15 & 0.214 & 0.075 \\ 
III & 19 & 0.283 & 0.088 & IV & 19 & 0.210 & 0.078 \\ \hline
\multicolumn{8}{c}{$n=200$}\\ 
Perturbation & Obs & Mean & SD & Perturbation & Obs & Mean & SD \\ \hline
I & 3 & 0.005 & 0.007 & II & 3 & 0.004 & 0.006 \\ 
I & 15 & 0.005 & 0.007 & II & 15 & 0.004 & 0.006 \\ 
I & 19 & 0.004 & 0.007 & II & 19 & 0.183 & 0.047 \\ 
III & 3 & 0.003 & 0.004 & IV & 3 & 0.133 & 0.035 \\ 
III & 15 & 0.154 & 0.040 & IV & 15 & 0.130 & 0.036 \\ 
III & 19 & 0.153 & 0.040 & IV & 19 & 0.133 & 0.038 \\ \hline
\end{tabular}
  \label{ci_srm}
\end{table}

\clearpage

\begin{table}
  \centering
  \caption{Case influence diagnostic for GARCH using normalizing Bregman divergence.}
  \begin{tabular}{cccc cccc}
\hline 
\multicolumn{8}{c}{$T=100$}\\
Perturbation & Obs & Mean & SD & Perturbation & Obs & Mean & SD \\ \hline
I & 19 & 0.009 & 0.014 & II & 19 & 0.007 & 0.010 \\ 
I & 44 & 0.009 & 0.014 & II & 44 & 0.007 & 0.011 \\ 
I & 64 & 0.009 & 0.014 & II & 64 & 0.247 & 0.072 \\ 
III & 19 & 0.005 & 0.006 & IV & 19 & 0.162 & 0.051 \\ 
III & 44 & 0.188 & 0.059 & IV & 44 & 0.156 & 0.048 \\ 
III & 64 & 0.192 & 0.060 & IV & 64 & 0.156 & 0.049 \\ \hline 
\multicolumn{8}{c}{$T=500$}\\
Perturbation & Obs & Mean & SD & Perturbation & Obs & Mean & SD \\ \hline
I & 119 & 0.001 & 0.002 & II & 119 & 0.001 & 0.002 \\ 
I & 344 & 0.002 & 0.004 & II & 344 & 0.001 & 0.003 \\ 
I & 464 & 0.001 & 0.003 & II & 464 & 0.063 & 0.029 \\ 
III & 119 & 0.001 & 0.002 & IV & 119 & 0.054 & 0.023 \\ 
III & 344 & 0.058 & 0.026 & IV & 344 & 0.054 & 0.025 \\ 
III & 464 & 0.059 & 0.028 & IV & 464 & 0.052 & 0.021 \\ \hline
\end{tabular}
\label{ci_garch}
\end{table}

\clearpage

\begin{figure}[h]\centering
  \includegraphics[width=7cm,angle=270]{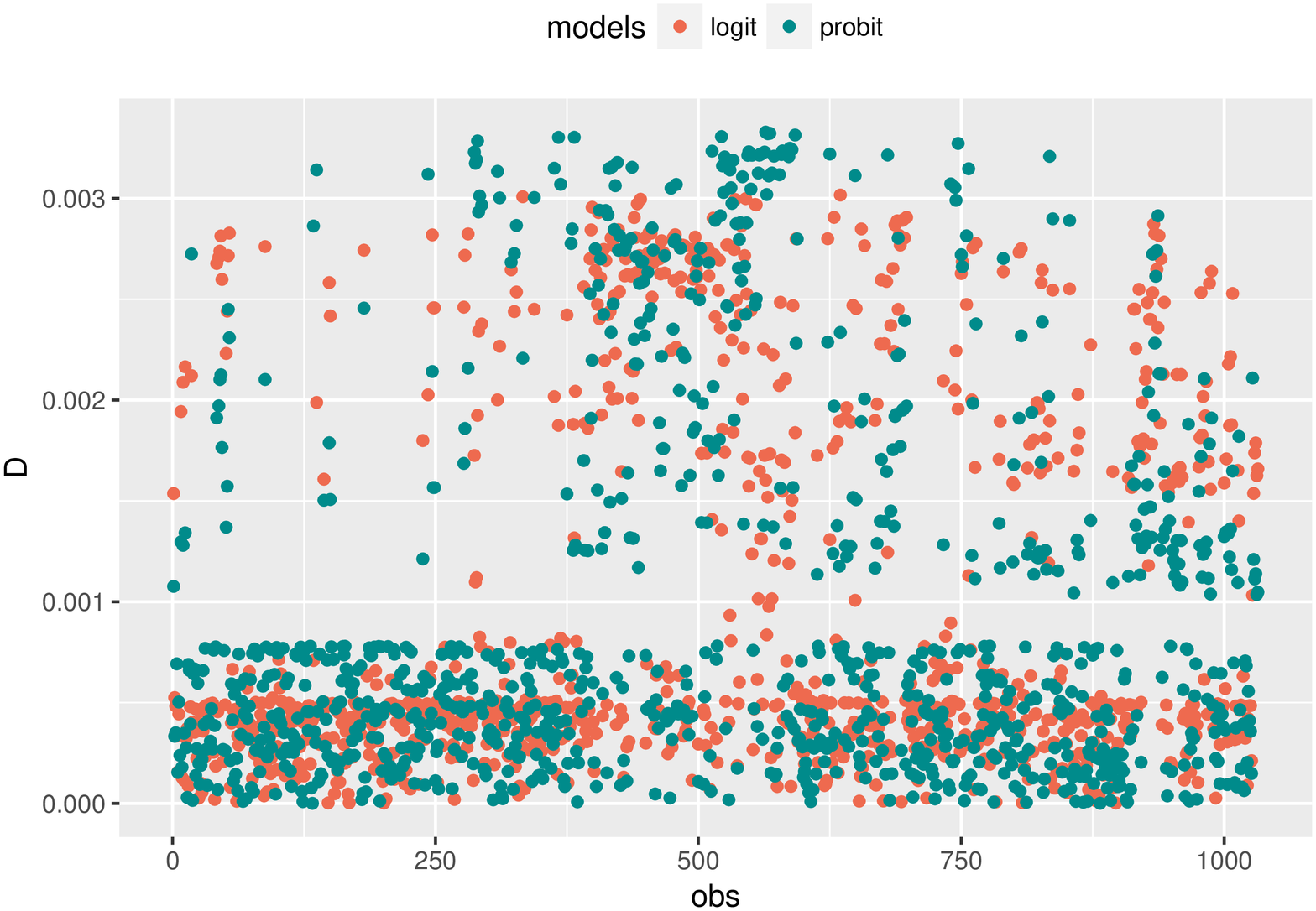}
  \caption{Normalizing Bregman divergence for an endemic coastal alpine bird (1990-2013).}
  \label{pombas_breg}
\end{figure}

\clearpage

\begin{figure}[h]\centering
  \includegraphics[width=12cm]{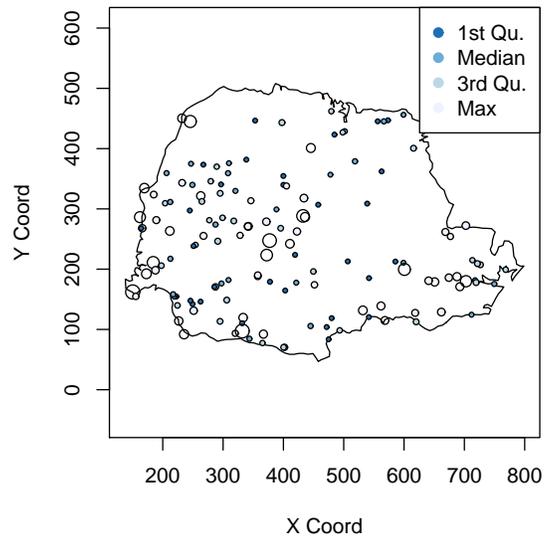}
  \caption{Normalizing Bregman divergence for average precipitation in Paran\'a, Brazil.}
  \label{parana_map_breg}
\end{figure}

\clearpage

\begin{figure}[h]\centering
  \includegraphics[width=8cm,angle=270]{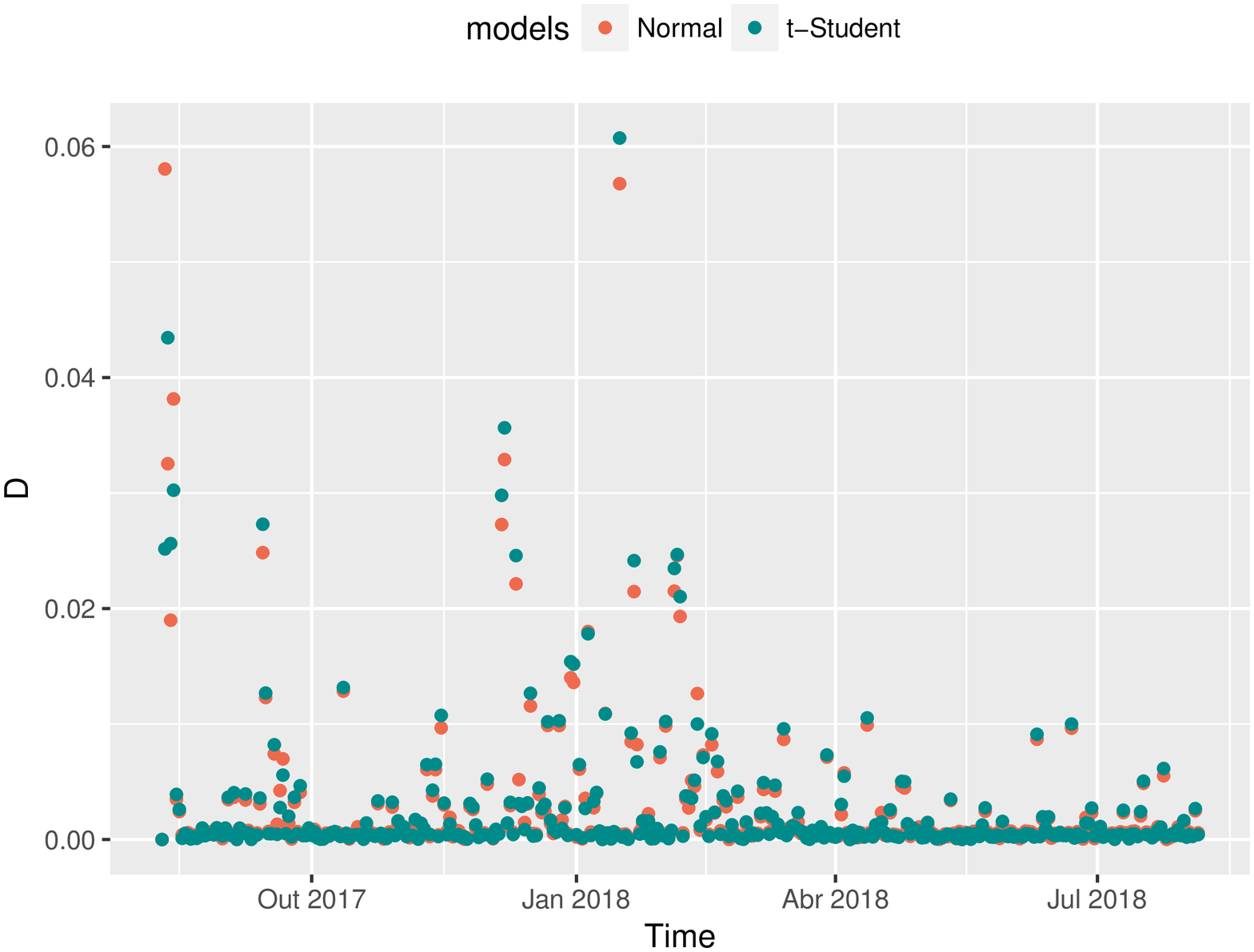}
  \caption{Normalizing Bregman divergence for Bitcoin exchange to US Dollar (August 5 2017 - August 5 2018).}
  \label{btc}
\end{figure}

\end{document}